\documentclass[manuscript]{acmart}

\newif\ifapx

\newif\ifgenplot
\genplotfalse 

\newcommand{\ourmaintitle}{Discovering Reliable Approximate Functional Dependencies}

\newcommand{\oururl}{\url{http://eda.mmci.uni-saarland.de/dora/}}

\usepackage{latexsym}
\usepackage[ruled, vlined, nofillcomment, linesnumbered]{algorithm2e}
\usepackage{eda} 
\usepackage{pifont}
\usepackage{balance}

\usepackage{adjustbox}
\usepackage[labelfont=sf, textfont=sf,singlelinecheck=off,justification=centering]{subcaption} 

\graphicspath{ {./} }

\ifpdf
\usetikzlibrary{external}
\fi

\SetKwComment{tcpas}{\{}{\}}
\SetCommentSty{textnormal}
\SetArgSty{textnormal}
\SetKw{False}{false}
\SetKw{True}{true}
\SetKw{Null}{null}
\SetKwInOut{Output}{output}
\SetKwInOut{Input}{input}
\SetKw{AND}{and}
\SetKw{OR}{or}
\SetKw{Continue}{continue}

\newcommand{\defemph}[1]{\textbf{#1}}
\newcommand{\mapping}[3]{#1 \! : #2 \to #3}
\newcommand{\card}[1]{| #1 |}
\newcommand{\with}{\! : \,}
\newcommand{\given}{\, | \,}

\newcommand{\mo}{\hat{m}_o\xspace}
\newcommand{\hI}{\ensuremath{\hat{I}}\xspace}
\newcommand{\cF}{\mathcal{F}}
\newcommand{\cI}{\mathcal{I}}
\newcommand{\cP}{\mathcal{P}}
\newcommand{\cR}{\mathcal{R}}
\newcommand{\cS}{\mathcal{S}}
\newcommand{\cT}{\mathcal{T}}
\newcommand{\cX}{\mathcal{X}}
\newcommand{\cZ}{\mathcal{Z}}

\newcommand{\biasn}{b(n)\xspace}

\newcommand{\expectind}[2]{\mathbb{E}_{#1}[#2]\xspace}

\newcommand{\expectn}[1]{\mathbb{E}[#1]\xspace}

\newcommand{\X}{\mathcal{X}}
\newcommand{\vd}{\mathbf{d}}
\newcommand{\vx}{\mathbf{x}}
\newcommand{\vv}{\mathbf{v}}

\newcommand{\NMI}{F\xspace}
\newcommand{\hNMI}{\ensuremath{\hat{F}}\xspace}
\newcommand{\AMI}{\ensuremath{\hat{F}_{\text{adj}}}\xspace}
\newcommand{\ourScore}{\ensuremath{\hat{F}_{0}}}

\newcommand{\MI}{I\xspace}

\DeclareMathOperator*{\argmax}{arg\,max}

\newcommand{\D}{\mathbf{D}}

\newcommand{\domain}[1]{V(#1)}
\newcommand{\vdom}[1]{\mathbf{V}(#1)}
\newcommand{\empdom}[1]{\hat{V}(#1)}
\newcommand{\empvdom}[1]{\hat{\mathbf{V}}(#1)}

\setcopyright{rightsretained}
\acmDOI{10.475/123_4}
\acmISBN{123-4567-24-567/08/06}
\acmConference[KDD'17]{ACM SIGKDD conference}{August 2017}{Halifax, Canada} 
\acmYear{2017}\copyrightyear{2017}\acmPrice{15.00}

\begin{document}
	\setlength{\pdfpagewidth}{8.5in}
	\setlength{\pdfpageheight}{11in}
	
	\title{\ourmaintitle}
	
	\author{Panagiotis Mandros}
	\affiliation{%
		\institution{Max Planck Institute for Informatics and Saarland University, Germany}
	}
	\email{pmandros@mpi-inf.mpg.de}
	
	\author{Mario Boley}
	\affiliation{%
		\institution{Max Planck Institute for Informatics and Saarland University, Germany}
	}
	\email{mboley@mpi-inf.mpg.de}
	
	\author{Jilles Vreeken}
	\affiliation{%
		\institution{Max Planck Institute for Informatics and Saarland University, Germany}
	}
	\email{jilles@mpi-inf.mpg.de}
	
	\begin{abstract}
		\label{sec:abstract}
		
		Given a database and a target attribute of interest,  how can we tell whether there exists a functional,  or approximately functional dependence of the target  on any set of other attributes in the data?  How can we reliably, without bias to sample size or dimensionality,  measure the strength of such a dependence?  And, how can we efficiently discover the optimal or  $\alpha$-approximate top-$k$ dependencies?  These are exactly the questions we answer in this paper.
		
		As we want to be agnostic on the form of the dependence,  we adopt an information-theoretic approach, and  construct a reliable, bias correcting score that  can be efficiently computed. Moreover, we give an effective  optimistic estimator of this score, by which for the first time we can mine the approximate functional dependencies  from data with guarantees of optimality.  Empirical evaluation shows that the derived score achieves a good bias for variance trade-off, can be used within an efficient discovery algorithm, and indeed discovers meaningful dependencies. Most important, it remains reliable in the face of data sparsity.
		
	\end{abstract}

	
	
	\maketitle
	
	\section{introduction} \label{sec:introduction}

	Discovering dependencies is an important and well-studied topic in data mining. 
	Most proposals, however, focus specifically on \emph{symmetric} dependencies. 
	That is, they aim to find variable sets  that strongly correlate or associate with a target variable.
	In many applications, however, \emph{asymmetric}, or targeted dependencies are 
	of particular interest. When anonymizing a dataset, for example, we need to be certain a private attribute cannot be reconstructed given the public attribute,
	while we do not care for the opposite direction.
	Similarly, in scientific applications we want to hypothesize whether
	a certain target variable, say an effect, can be explained by
	the observed variables, the potential causes, and not the other way around.
	Generally, an effective procedure to detect functional dependencies 
	from data allows us to rule out alternate theories about our domain and 
	to determine whether finding concrete models, e.g., by statistical learning, 
	is worthwhile, or if we rather should acquire more data 
	or enrich our feature space first~\cite{ghiringhelli2015big}.
	
	More formally, given a target variable $Y$ and a set of attributes $\X$, 
	we want to \emph{measure} the degree at which $Y$ has a functional, 
	or an \emph{approximate} functional dependence on $\X$, i.e., if $Y \approx f(\X)$. 
	Additionally, we want to efficiently discover whether any such $\X$ 
	exists in our data. The database community studied how to infer \emph{exact} 
	functional dependencies, as these allow for normalization, i.e., reducing redundancy. These methods are not suited to our end, however, 
	as they do not measure the approximation in terms of an intuitive score, and in addition, make implicit closed-world assumptions based on the 
	schema of the data~\cite{huhtala1999tane,liu2012discover,giannella:2004:fraction}. 
	
	On the contrary, information theory provides an intuitive and interpretable measure to address these issues. The fraction of information quantifies functional dependence in terms of proportional reduction of uncertainty about $Y$ when observing  $\X$~\cite{cavallo:1987:fraction,dalkilic:2000:ind,reimherr:2013:fraction}. Information-theoretic measures, however, are sensitive to data sparsity and as a result, the fraction of information overestimates the amount of dependence~\cite{romano:2016:chance}.  For large dimensionalities of $\X$,  it is even possible that a functional dependence is indicated when $\X$ and $Y$ are actually independent. This makes it a non-\emph{reliable} score. In addition, maximizing it is {\bf NP}-hard \cite{krause2012near}.

	In this paper we propose a \emph{reliable} measure for approximate functional dependencies
	based on the fraction of information. Even in extreme cases of data sparsity, it does not show dependence. In addition, we derive an effective optimistic estimator for this score, that allows for an admissible branch-and-bound algorithm to discover the top-$k$ optimal, or $\alpha$-approximate optimal strongest dependencies. Empirical evaluation shows that the derived score achieves a good bias for variance trade-off, and in addition, it does not favor spurious dependencies. The corresponding optimistic estimator is a data-dependent quantity, and by exploiting the structure of the data, leads to an effective search algorithm. Lastly, concrete findings in two exemplary application domains, AI and Materials Science, reproduce sensible domain information.
	
	The main contributions of this paper are the following. We
	
	\begin{itemize}[noitemsep,topsep=2pt]
		\item[(i)] propose a consistent estimator for the fraction for information score that is not prone to spurious dependencies,
		\item[(ii)] provide an efficient branch-and-bound algorithm for the discovery of optimal, and $\alpha$-approximate optimal top-$k$ dependencies, and
		\item[(iii)] provide empirical evaluation on a wide range of real and synthetic datasets.
	\end{itemize}
	
	The paper is structured as follows. We formally introduce the two problems we consider in Section~\ref{sec:problem}. Next, in Section~\ref{sec:score} we propose our fraction of information score, and in Section~\ref{sec:search} we detail how to optimize it by deriving a bounding function for a branch-and-bound search scheme. Following, in Section~\ref{sec:exps} we evaluate the performance on a variety of tasks. Finally, we round up with conclusions in Section~\ref{sec:concl}.
	
	\section{Problem Definition} \label{sec:problem}
	We consider a discrete sample space governed by some probability mass function $p$ for which we have defined $d+1$ discrete random variables $\cR = \lbrace X_1,\dots,X_d,Y\rbrace$ with domains $\domain{X_1},\dots,\domain{X_d}$, and $\domain{Y}$, respectively.
	Subsets $\cS \subseteq \cR$ are identified with vector-valued random variables in the usual way with domain $\vdom{\cS} = \bigtimes_{R \in \cS} \domain{R}$.
	We consider the variable $Y$ as the \defemph{output variable} and the remaining variables $\cI=\lbrace X_1,\dots,X_d \rbrace$ as the \defemph{input variables}, and our goal is to discover subsets of the input variables $\X \subseteq \cI$ that approximately \emph{determine} $Y$.
	In particular, we are interested in approximations to the concept of \defemph{functional dependencies}, i.e., the case when there is a function $\mapping{f}{\vdom{\X}}{\domain{Y}}$ such that for all $\vx \in \vdom{\X}$ it holds that 
	\begin{equation}
		p(Y=y \mid \X=\vx)=
		\begin{cases}
			\; \; 1 &, \text{ if } y=f(\vx)\\
			\; \; 0 &, \text{ otherwise}
		\end{cases} \enspace .
		\label{eq:fd}
	\end{equation}
	Relaxing this rather strict concept is necessary because it is rare that such a completely deterministic relationship exists---if the random variables correspond to measurements of real-world quantities there are usually unobserved subtle effects or noise that cause Eq.~\eqref{eq:fd} to not hold exactly. 
	
	One traditional approach to relax Eq.~\eqref{eq:fd} is to use instead the condition $p(Y=y \mid \X=\vx) \geq 1-\epsilon$ if $y=f(\vx)$, for some fixed value $\epsilon \in (0,1]$, i.e., to allow a certain fraction of events to not obey the functional relation.
	However, as with any parameterization based on a hard threshold, this parameter is difficult to set in practice and additionally only provides a qualitative and not a quantitative relaxation \cite{giannella:2004:fraction}. That is, it does not allow us to express ``how far'' is $Y$ from being determined by $\X$. In order to address these issues one can quantify the degree of functional dependence through information theoretic measures.
	A particularly useful way of doing this  is to use the concept of \defemph{fraction of information} ($\NMI$)~\cite{cavallo:1987:fraction, reimherr:2013:fraction,dalkilic:2000:ind}, which is defined as
	\begin{align}
		\NMI (\X ; Y) = \frac{H(Y)-H(Y \mid \X)}{H(Y)} 
	\end{align}
	where $H(Y) = -\sum_{y \in \domain{Y}} p(y)\log(p(y))$ denotes the \defemph{Shannon entropy} and $H(Y \mid \X)=\sum_{\vx \in \vdom{\X}} p(\vx) H(Y \mid \X=\vx)$ the conditional Shannon entropy~\cite{shannon:1948:communication}. The numerator is referred to as \defemph{mutual information} $\MI(\X;Y)=H(Y)-H(Y \mid \X)$. The entropy measures the uncertainty about $Y$, while the conditional entropy measures the uncertainty about $Y$ after observing $\X$. The fraction of information then represents the proportional reduction of uncertainty about $Y$ by knowing $\X$. 
	Moreover, the extreme values, $\NMI(\X;Y)=1$ and $\NMI(\X;Y)=0$, correspond to a functional dependence and statistical independence, respectively. With this notion, we can go about discovering approximate functional dependencies from  data.
	
\begin{figure}[t]
	\centering
	\begin{minipage}[b]{0.75\linewidth}
		\ifgenplot
		\ifpdf
		\tikzsetnextfilename{figure1}
		\fi
		\begin{tikzpicture}
		\begin{axis}[eda line, xlabel={Dimensionality of $\cX$},ylabel={Fraction of Information}, ymin=0.0, ymax=0.75, height=5cm, width=\linewidth, xmin=1.0, xmax=5.0, legend style={nodes={scale=0.8, transform shape}, at={(0.06,1.0)}, anchor=north}, ytick={0.0,0.25, ...,1.0}]
		\addplot+[mambacolor1] table[x index=0, y index=2, header=true] {../expres/cardinality_vs_independence.dat};
		\addplot+[mambacolor3] table[x index=0, y index=6, header=true] {../expres/cardinality_vs_independence.dat};
		\legend{\ourScore, \hNMI}
		\end{axis}
		\end{tikzpicture} 	 	
		\else
		\includegraphics{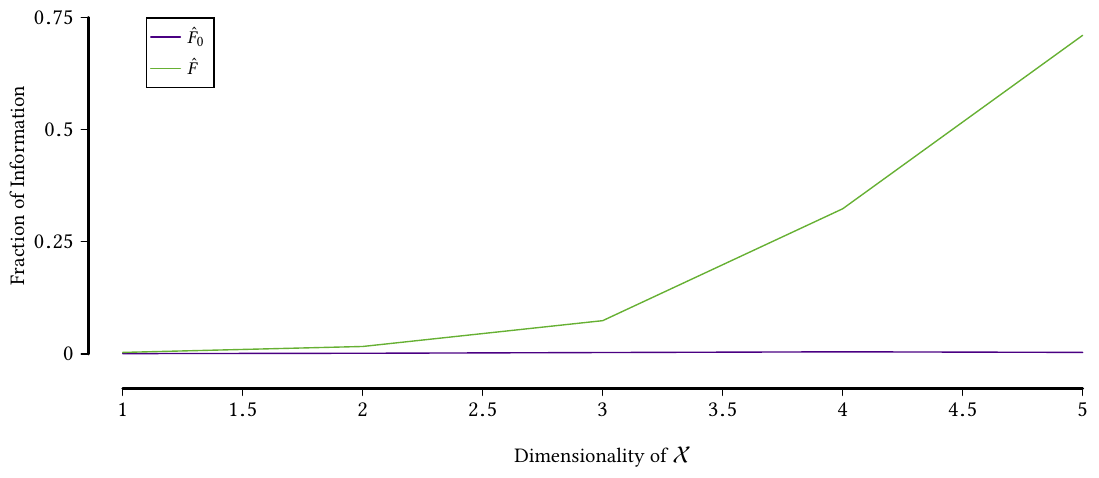}
		\fi 	
	\end{minipage}
		\caption{The fraction of information score against increasing dimensionality for $\X$; using $n=1000$ samples from $\cR=\{X_1,\dots,X_5,Y\}$ where all variables are mutually independent and $\card{\domain{A}}=4$ for all $A \in \cR$.
			Since $\cX$ and $Y$ are independent, the reliable fraction of information should be constantly 0. However, the baseline estimator $\hNMI$ shows increasing functional dependence. On the contrary, our proposed corrected score,  $\ourScore$, is always 0.}
		\label{fig:increasingDomain}
	\end{figure}
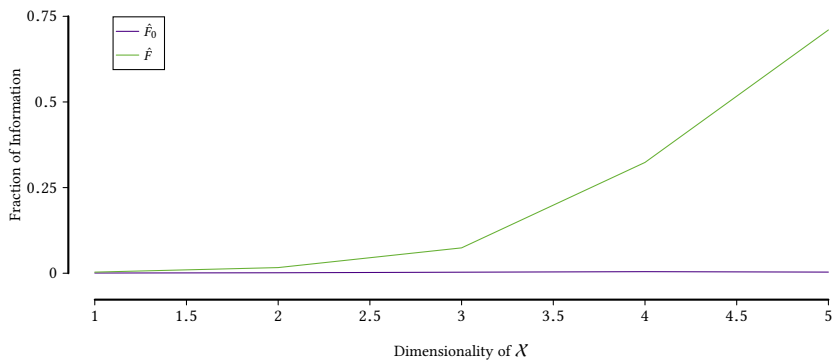
	
	For that we assume that a dataset $\D_n \in \vdom{\cR}^n$ is given consisting of $n$ i.i.d. samples $\vd_1,\dots,\vd_n$ generated according to the joint distribution $p$. 
	Such a dataset induces empirical probability estimates for all our random variables $\cS \subseteq \cR$ given by
	$\hat{p}(\cS=\vv)=c(\cS=\vv)/n$ with the \defemph{empirical counts}
	\[
	c(\cS=\vv)=\card{\lbrace \vd \in \D_n \with \vd(S)=\vv(S) \text{ for all } S \in \cS \rbrace}
	\]
	(where $\vd(S)$ is the entry in $\vd$ corresponding to variable $S$).
	In turn, these empirical probabilities give rise to empirical estimators $\hat{H}, \hat{\MI}, \hat{\NMI}$ for our quantities of interest, $H$, $\MI$, and $\NMI$.
	However, trying to directly discover approximate functional dependence using the empirical fraction of information $\hat{\NMI}$ is bound to fail, because this estimator is not unbiased, i.e., we have $\expectn{\hat{\NMI}(\X;Y)} \neq \NMI(\X;Y)$ for finite $n$, as it is the case with many dependence measures \cite{nguyen:2010:chancejournal,romano:2014:smi,hubert:1985:rand,dobra:2001:bias,white:1994:bias,kononenko:1995:bias,uds,wang:2017:bias}.
	This holds in particular also for the case when $\NMI(\X;Y)=0$, i.e., when $\X$ contains no information about $Y$. This situation, which is referred to as the (lack of) zero-baseline property \cite{romano:2016:chance}, can be misleading in practice.
	Even worse, independent of the true value $\NMI(\X;Y)$, the bias depends on the number of attainable distinct values for $\X$, and favors larger attribute sets over smaller ones (which follows from the bias of the empirical mutual information, see, e.g., \cite{roulston1999estimating}). See Fig.~\ref{fig:increasingDomain} for a quantitative demonstration of both of these facts.
	
	Even if a more suitable estimator was available, the challenge of which variable sets $\X \subseteq \cI$ to test for high functional dependence scores remains---naively considering all $2^d$ options is practically infeasible.
	Thus, to derive a useful method for the reliable discovery of functional dependencies from data, we have to solve the following two problems:
	\begin{itemize}
		\item[(i)] Find a more reliable empirical estimator $\hat{\NMI}'$ for $\NMI$; in particular one that satisfies the zero-baseline property and obtains better dimensionality bias.
		\item[(ii)] Identify structural properties of $\hat{\NMI}'$ that allow to derive an effective search algorithm for discovering the variable sets with the highest functional dependence scores.
	\end{itemize}
	We will present solutions to each of these problems in turn, in Sections~\ref{sec:score} and \ref{sec:search}, respectively.

	\section{Reliable fraction of information}
	\label{sec:score}
	Intuitively, the reason why $\hat{\NMI}$ is unreliable as an estimator for $\NMI$ is that it does not take into account the confidence in the empirical estimates 
	$\hat{H}(Y|\X=\vx)$. This is especially obvious in the extreme case when the empirical count $c(\X=\vx)$ is equal to $1$. 
	In this situation $c(Y=y, \X=x)=1$ exactly for one value of $y \in \domain{Y}$ and, hence, $\hat{H}(Y|\X=\vx)$ is trivially equal to $0$ independent of the true distribution $p$. Moreover, if the data size $n$ is small compared to the observed domain of $\cX$, this case is likely to occur for many of the sampled values for $\cX$---even when $\NMI(\X;Y)=0$, which coincides with the highest error, because then $H(Y|\X=\vx)=H(Y)$.
	
	This last observation suggests a path to a more reasonable estimator: 
	while it is hard to determine the bias in the general case, it is likely much easier under the assumption of independence $\NMI(X;Y)=0$, which, as pointed out above, corresponds exactly to the case of highest estimation error when the empirical observations are sparse.
	More concretely, let us denote by $b_0(\cX,Y,n)$ the \defemph{bias under independence} defined as
	\[
	b_0(\cX,Y,n)=\expectn{\hat{\NMI}(\X;Y) \given \NMI(X;Y)=0} 
	\]
	where the expectation is taken w.r.t. the random dataset $\D_n \sim p $ of size $n$. Let us assume that we have a good estimator $\hat{b}_0$ for this quantity. With this we can define a corrected estimator, let us refer to it as \defemph{reliable fraction of information} $\ourScore$, as follows:
	\[
	\ourScore(\X ; Y) =  \hat{\NMI}(\X;Y)-\hat{b}_0(\X,Y,n) \enspace .
	\]
	This approach essentially trades the bias of $\hat{\NMI}$ with that of $\hat{b}_0$ when $\NMI(X;Y)=0$. We have:
	\begin{align*}
		& \expectn{\ourScore(\X;Y) - \NMI(X;Y) \given \NMI(X;Y)=0}\\
		=& \expectn{\hat{\NMI}(\X;Y)-\hat{b}_0(\X,Y,n)-0  \given \NMI(\X;Y)=0}\\
		=& b_0(\X,Y,n) - \expectn{\hat{b}_0(\X,Y,n) \given \NMI(\X;Y)=0} \enspace .
	\end{align*}
	
	\newcommand{\epm}[1]{\hat{\mathbb{E}}_0[#1]}
	\newcommand{\ppm}[1]{\hat{\mathbb{P}}_0[#1]}
	A non-parametric choice for $\hat{b}_0$, which we use in this paper, corresponds to the \emph{permutation model}~(\citet[Chap. 11.2]{lancaster:1969:chi}), i.e., considering all possible datasets $\D'_n$ resulting from independently permuting the $Y$-values associated to the $\X$-values in the given empirical data $\D_n$. Formally, let $S_n$ denote the symmetric group of degree $n$, i.e., $S_n$ consists of all $n!$ bijections $\mapping{\sigma}{\{1, \dots, n\}}{\{1, \dots, n\}}$.
	For a bijection $\sigma \in S_n$, let $Y_\sigma$ denote the permuted version of $Y$, i.e., the variable with data entries $\vd_i(Y_\sigma)=\vd_{\sigma(i)}(Y)$.
	With this we can define the \defemph{permutation model} as the probabilities $\hat{\mathbb{P}}_0$ (and corresponding expectations $\hat{\mathbb{E}}_0$) resulting from permuting the empirical data of $Y$ by a uniform random permutation from $S_n$.
	Using this model, the expectation of the empirical mutual information under independence, $\mo$, is estimated as
	\begin{align}
		\mo(\cX,Y,n)  =  \epm{\hI (\cX,Y_\sigma)}= \frac{1}{n!} \sum_{\sigma \in S_n} \hI (\cX,Y_\sigma) \enspace .
	\end{align}
	
	Clearly, a naive evaluation of this expression is computationally infeasible (order of $n!$). However, one can dramatically reduce the complexity by reformulating the above expression as a function of contingency table cell values and exploiting its symmetries~\cite{nguyen:2009:ami}.
	More precisely, let the observed domains of $\X$ and $Y$ be $\empvdom{\X}=\lbrace \vx_1,\dots, \vx_R \rbrace$ and $\empdom{Y}=\lbrace y_1,\dots, y_C \rbrace$, respectively.
	Moreover, we define shortcuts for the observed marginal counts $a_i=c(\X=\vx_i)$ and $b_j=c(Y=y_j)$ as well as for the joint counts $c_{i,j}=c(\X=\vx_i, Y=y_j)$.
	The complete joint count configuration $c=(c_{i,j} \with 1\leq i \leq C, 1 \leq j \leq R)$ we refer to as \defemph{contingency table}.
	Noting that $\hat{I}(\X,Y_\sigma)$ is a function of the contingency table $c$ resulting from the random permutation, the estimator $\hat{m}_0$ can be rewritten as
	\begin{equation}
		\hat{m}_0(\X,Y,n)=\sum_{c \in \mathcal{T}}\ppm{c} \hat{I}(c) = \sum_{c \in \mathcal{T}}\ppm{c}\sum_{i=1}^{R} \sum_{j=1}^{C}\frac{c_{ij}}{n} \log \frac{ c_{ij} n} {a_i b_j}
		\label{eq:permodel}
	\end{equation}
	where $\cT=\cT(\X,Y)$ is the set of all possible contingency tables indexed by the values $\empvdom{\X}$ and $\empdom{Y}$ (note that $\ppm{c}>0$ only for $c$ resulting in the observed marginal counts $a,b$).
	
	As this expression is still infeasible, \citet{nguyen:2009:ami} propose to re-order the terms of the sum according to the possible count values that can be found in individual table cells.
	The permutation model implies that the empirical counts $c_{ij}$ for the joint events $\X=\vx_i, Y=y_j$ are generated according to the probabilities 
	\[
	\ppm{c_{ij}=k}=h(k;a_i,b_j,n)
	\]
	where $h$ is the probability mass function of the hypergeometric distribution with $c_{ij}$ the number of successes, $a_i$ the number of draws, $b_j$ the number of total successes, and $n$ the population size. 
	This allows us to group terms according to their count values for a specific table cell, which can be systematically enumerated from the support of the hypergeometric distribution, i.e., $k \in [\max(0,a_i+b_j-n),\min(a_i,b_j)]$. We can then compute $\hat{m}_0$ as
	\[
	\hat{m}_0(\X,Y,n) =  \sum_{i=1}^{R} \sum_{j=1}^{C} \sum_{k=\max(0,a_i + b_j-n)}^{\min(a_i,b_j)} h(k ; a_i ,b_j,n) \frac{k}{n} \log \frac{kn} {a_i b_j} \enspace .
	\]
	Using the recurrence relation of the hypergeometric distribution, the computational complexity can be further reduced to the order of $\max(nR,nC)$~\cite{romano:2014:smi}. Moreover, it is easily parallelizable. 
	Hence, we end up with an efficiently computable estimator for the bias under independence $\hat{b}_0(\cX,Y,n)=\hat{m}_0(\X,Y,n) / \hat{H}(Y)$. 
	
	In addition to being computationally efficient, the resulting reliable functional dependence score $\ourScore(\X;Y)=\hat{\NMI}(\X;Y)-\hat{b}_0(\cX,Y,n)$ satisfies several other properties.
	First of all, it is indeed a consistent estimator of $\NMI$.
	One can show \citep{nguyen:2010:chancejournal} that $\lim_{n \rightarrow \infty} \hat{m}_0(\X,Y,n)=0$, which implies together with the consistency of $\hat{\NMI}$ that 
	\[
	\lim_{n \rightarrow \infty} \ourScore(\X;Y) = \NMI(\X;Y) \enspace .
	\]
	Moreover, $\ourScore$ remains upper-bounded by $1$, although this value is only attainable in the limit case $n \to \infty$ (for true functional dependencies). 
	Most importantly, in contrast to the naive estimator, $\ourScore$ approaches zero\footnote{It fact, it is principally not lower bounded by $0$ since the empirical fraction of information can be less than the correction term. However, these are rare cases, which  strongly indicate independence.} as the data size relative to the empirical domain $\empvdom{\X}$ approaches one. 
	In other words, $\hat{b}_0(\cX,Y,n)$ penalizes spurious dependencies that can easily appear for high dimensional $\X$---justifying the name \emph{reliable} fraction of information.
	
	\section{Search scheme}
	\label{sec:search}
	\newcommand{\oest}{\overline{f}}
	\newcommand{\branch}{\mathbf{r}}
	After deriving a suitably corrected empirical estimator for the fraction of information, we can now turn to the problem of using it for the discovery of approximate functional dependencies from a given dataset.
	Essentially, this is a combinatorial optimization problem where, given a dataset $\D_n$, the problem is to find a subset $\X \subseteq \cI$ with maximal value of $\ourScore$.
	However, as usual in pattern discovery, we are not just interested in one but several score maximizers---to produce more diverse insights into the data domain and to provide alternatives in subsequent applications.
	Hence, we end up with a top-$k$ pattern search formulation: 
	
	\textit{{\bf Given} a dataset $\D_n$ consisting of $n$ i.i.d. samples of random variables $\cI$ and $Y$ and a number $k$, 
		{\bf find} a family $\cF_k$ of variable sets $\X_1,\dots,\X_k \subseteq \cI$  such that no variable set outside $\cF_k$ has a higher $\ourScore$-score than any of the sets in $\cF_k$, i.e., for all $\X \in \cF_k$ and $\cZ \in \mathcal{P}({\cI})\setminus \cF_k$ it holds that 
		$\ourScore(\cZ;Y) \leq \ourScore(\cX;Y)$.
	}
	
	\DontPrintSemicolon
	\SetKwFunction{bstbb}{Bst-BB}
	\SetKwFunction{qempty}{empty}
	\SetKwFunction{qaddall}{addAll}
	\SetKwFunction{qtop}{top}
	\SetKwFunction{topk}{top-k}
	\SetKwFunction{suc}{suc}
	\SetKwFunction{argmax}{argmax}
	\begin{algorithm}[t]
		\bstbb{$\mathbf{Q}$,$\cF_k$,$k$}: 
		\Begin{
			\eIf{$\mathbf{Q}=\emptyset$ {\bf or} $\oest(\qtop{$\mathbf{Q}$})/f(\cF_k[k]) \leq \alpha$}{\KwRet{$\cF_k$}\;}
			{
				$\mathbf{R}=\branch(\qtop{$\mathbf{Q}$})$\; 
				$\cF'_k = \topk{$\cF_k \cup \mathbf{R}$}$\;
				$\mathbf{Q}'=(\mathbf{Q} \setminus \{\qtop{$\mathbf{Q}$}\}) \cup \{\cX \in \mathbf{R} \with \oest(\cX)/ f(\cF'_k[k])\geq \alpha\}$\;
				\KwRet{\bstbb{$\mathbf{Q}'$,$\cF_k$,$k$}}\;
			}
		}
		\BlankLine
		$\cF_k$ = \bstbb{$\{\bot\},\bot,k$} 
		\caption{Best-first branch-and-bound; Given input and output variables $\cI$ and $Y$, finds $\alpha$-approximation to top-$k$ result set $\cF_k \subseteq \cP(\cI)$ w.r.t. $f(\X)=\ourScore(\X;Y)$ using optimistic estimator $\oest(\X)=1-\hat{b}_0(\X;Y,n)$.}
		\label{alg:bb}
	\end{algorithm}
	In the search for an algorithm solving the above problem, it is first important to note that maximizing mutual information is {\bf NP}-hard~\cite{krause2012near}---even approximately and even in the restricted case when all $\X, \cZ \subseteq \cI$ are conditionally independent given $Y$.
	While it is an open question whether this result implies hardness of \ourScore -maximization (the correction term changes maximization order), this is a substantial indication that no polynomial time algorithm for our problem exists (even with constant $k$).
	On the other hand, the branch-and-bound framework (see, e.g., \citet[Chap. 12.4]{mehlhorn2008algorithms}), while not efficient in terms of the worst-case complexity, can often yield algorithms for hard optimization problems that are very effective in practice---particularly in the best-first search variant.
	
	In a nutshell, best-first branch-and-bound maximizes an objective function $\mapping{f}{\Omega}{\mathbb{R}}$ defined on some abstract search space $\Omega$ with the help of a \defemph{branch operator} $\mapping{\branch}{\mathcal{P}(\Omega)}{\mathcal{P}(\Omega)}$ and a matching auxiliary selection and \defemph{bounding function} $\mapping{\oest}{\Omega}{\mathbb{R}}$.
	The role of the branch operator is to non-redundantly generate the search space from some designated root element $\bot \in \Omega$, i.e., for all $\omega \in \Omega$ there must be a unique sequence $\bot=\omega_1,\dots,\omega_l=\omega$ such that $\omega_{i+1} \in \branch(\omega_i)$ for $i=1,\dots,l-1$.
	The bounding function must guarantee the property 
	\[
	\oest(\omega) \geq \max \lbrace f(\omega') \with \omega' \in \branch^*(\omega) \rbrace
	\]
	where $\branch^*(\omega)$ denotes the set of all $\omega' \in \Omega$ that can be generated from $\omega$ by multiple applications of $\branch$.
	Based on these ingredients, a branch-and-bound algorithm simply enumerates $\Omega$ starting from $\bot$, but uses $\oest$ to avoid expanding elements that cannot yield an improvement over the best solution found so far.
	
	For our problem the search space is $\mathcal{P}(\cI)$, for which a suitable branch operator is simply given by
	\[
	\branch(\X)=\{\X \cup \{X_i\} \with i > \max\{j \with X_j \in \X\}\}
	\]
	i.e., we ensure non-redundant generation by creating a lexicographical order on the power set of the input variables and only enumerate lexicographically larger elements from a given set $\X$. In order to derive a bounding function $\oest$ for our objective function $f(\cX)=\ourScore(\X;Y)$, we first need to
	establish another central property of the correction term $\hat{b}_0(\X,Y,n)$.
	\begin{theorem}
		\label{monononicityProof}
		Given two sets of variables $\X \subset \X' \subseteq \cI$ then $\hat{b}_0(\X,Y,n) \leq \hat{b}_0(\X',Y,n)$, i.e., the correction term is monotonically increasing with respect to the subset relation.
	\end{theorem}
	\begin{proof}
		It is sufficient to consider the case $\X'=\X \cup \{X\}$ for some $X \not \in \X$, i.e. the cardinality differs by one variable. The general case follows inductively. Following the notation of Eq.~\eqref{eq:permodel}, let the marginals of $Y,\X$, and $\X'$ be $a_i$ for $i=1, \hdots, R$, $b_j$ for $j=1, \dots, C$, and $b'_j$ for $j=1, \hdots, C'$, respectively. Note that $C' \geq C$.
		We need to show that $\hat{m}_0(\X,Y,n) \leq \hat{m}_0(\X',Y,n)$, i.e., per definition $\sum_{c \in \mathcal{T}} \ppm{c} \hat{\MI}(c) \leq \sum_{c' \in \mathcal{T'}} \ppm{c'}\hat{\MI}(c')$.
		
		In order to do this, we first define a relation between the contingency tables of $\cT=\cT(\X,Y)$ and $\cT'=\cT(\X',Y)$.
		Let $\mapping{\pi}{\empvdom{\cX'}}{\empvdom{\cX}}$ be the projection of values from $\X'$ to values of $\X$ defined by $\pi(\vx')=\vx$. We can extend this projection to the sets of contingency tables $\mapping{\pi}{\mathcal{T}'}{\mathcal{T}}$ by finding the counts in the column corresponding to $\vx \in \vdom{X}$ of $\pi(c')$ as the sum of the columns in $c'$ corresponding to $\pi^{-1}(\vx)$. We will prove the claim by showing that for all $c \in \cT$ we have $\ppm{c} \hat{\MI}(c) \leq \sum_{c' \in \pi^{-1}(c)}\ppm{c'}\hat{\MI}(c')$.
		
		First, it follows from the chain rule of information
		and from mutual information being non-negative~\cite{cover} that $\hat{\MI}(c) \leq \hat{\MI}(c')$ for $c=\pi(c')$. Next we show that $\ppm{c} = \sum_{c' \in \pi^{-1}(c)}\ppm{c'}$, which concludes the proof.
		For any contingency table $z \in \cT(\cZ,Y)$ let $S_n[z]=\{\sigma \in S_n \with c(\cZ,Y_\sigma)=c\}$ denote the set of permutations that result in $z$.
		Let $\sigma \in S_n\setminus S_n[c]$. This means that $c_{i,j}(\X,Y) \neq c_{i,j}(\X,Y_\sigma)$ for at least one cell $i,j$.
		Denoting by 
		\[
		\pi^{-1}(j)=\{j' \with 1 \leq j' \leq C', \pi(\vx'_{j'})=\vx_j\}
		\]
		the set of all indices of values of $\cX'$ that are projected down to $\vx$, it follows
		by the definition of $\pi$ that
		\[
		\sum_{j' \in \pi^{-1}(j)} c'_{i,j'}(\cX',Y) \neq \sum_{j' \in \pi^{-1}(j)} c'_{i,j'}(\cX',Y_\sigma) \enspace .
		\]
		So for at least one $j' \in \pi^{-1}(j)$ it is $c'_{i,j'}(\cX',Y) \neq c'_{i,j'}(\cX',Y_\sigma)$, and, thus we also have that $\sigma \not\in S_n[c']$ and can conclude
		\begin{equation}
			\label{eq:supseteq343}
			S_n[c] \supseteq \bigcup_{c' \in \pi^{-1}(c)} S_n[c'] \enspace .
		\end{equation}
		Now let $z' \in \cT(\X',Y)$ with $\pi(z') \neq c$ and assume for a contradiction that $S_n[c] \supset S_n[c'] $, i.e., there is an $\sigma \in S_n[c] \cap S_n[z']$.
		Let us denote $z=\pi(z')$.
		Since $S_n[c] \cap S_n[z]=\emptyset$, we know that $\sigma \not\in S_n[z]$.
		However, it follows from Eq.~\eqref{eq:supseteq343} that $\sigma \not\in S_n[z']$---a contradiction, and, hence $S_n[c] = \bigcup_{c' \in \pi^{-1}(c)} S_n[c']$.
		Thus, as desired 
		\[
		\ppm{c}=\frac{\card{S_n(c)}}{\card{S_n}}=\sum_{c' \in \pi^{-1}(c)} \frac{\card{S_n(c')}}{\card{S_n}} = \sum_{c' \in \pi^{-1}(c)}\ppm{c'} \enspace .
		\]
	\end{proof}
	
	With this theorem (and the fact that the conditional entropy is bounded from below by 0), it follows for all $\mathcal{X} \subseteq  \mathcal{X}' \subseteq \mathcal{I}$ that	\begin{align*}
		\ourScore(\X';Y) = & \frac{\hat{H}(Y)-\hat{H}(Y \mid \X' )}{\hat{H}(Y)}  - \hat{b}_0(\X',Y,n) \\   
		\leq & 1-\hat{b}_0(\X,Y,n)= \ourScore'(\X;Y)
	\end{align*}
	Hence, since $\X \subseteq \X'$ for $\X' \in \branch(\X)$, we can use $\oest(\cX)=1-\hat{b}_0(\X,Y,n)$ as valid bounding function for the branch-and-bound search. Besides the features mentioned above, the search scheme also provides the option of relaxing the required result guarantee to that of an $\alpha$-approximation for accuracy parameter $\alpha \in (0,1]$.
	This means that the resulting family of variable sets $\cF_k$ will satisfy the relaxed condition that for all $\X \in \cF_k$ and $\cZ \in \mathcal{P}({\cI})\setminus \cF_k$ it holds that $\alpha \ourScore(\cZ;Y) \leq \ourScore(\cX;Y)$.
	Hence, using $\alpha$-values of less than $1$ allows to trade accuracy for computation time. 
	
	The pseudocode in Algorithm~\ref{alg:bb} summarizes the resulting method for the discovery of approximate functional dependencies. The algorithm maintains a priority queue $\mathbf{Q}$ that holds the search frontier and a current result set $\cF_k$ throughout the search. In the beginning of each iteration, it checks whether the search has terminated, which is the case either when the frontier is empty or the potential of top-potential element from the frontier (w.r.t. the bounding function) is less than the $k$-th best $\ourScore$-value of the current result. As long as this condition is not satisfied, the search continuous by expanding the top-potential variable set (line 5), and using its successors to update the current result set (line 6), as well as the priority queue (line 7).
	
	\section{Empirical Evaluation}
	\label{sec:exps}
	\newcommand{\hoest}{\overline{h}}
	
	In this section, we study the empirical performance of discovering approximate functional dependencies based on $\ourScore$.
	This includes, the bias of $\ourScore$ as an estimator of the true $\NMI$ functional, the performance of the bounding function $\oest$ in branch-and-bound search, and two concrete examples of functional dependencies in real datasets.
	
	\subsection{Estimation Bias}
	In this section we evaluate the bias and variance of our corrected estimator $\ourScore$. It is instructive to see the behavior of the bias for various amounts of dependence, and not in the particular case of independence, i.e., $\NMI=0$, where $\ourScore$ aims to be unbiased (see Fig.~\ref{fig:increasingDomain} for an empirical confirmation of this fact).   For that let us denote by $P$ the set of all joint probability mass functions over two random variables $X$ and $Y$ with $\card{\domain{X}}=\card{\domain{Y}}=3$, and by $P[a,b]$ all such probability mass functions for which we have a functional dependence score of $\NMI_p(X;Y) \in (a,b]$. We are interested in the behavior of the estimation bias over $P$ under a distribution that puts equal weight on the four different regimes ``weak'' ($P[0,0.25]$), ``low'' ($P[0.25,0.5]$), ``high'' ($P[0.5,0.75]$), and ``strong'' ($P[0.75,1]$).
	
	More specifically, let $\tau(\D_n)$ be the result of the $\NMI$-estimator $\tau$ computed on data $\D_n$ and $\biasn(p,\tau)$ the bias of $\tau$ when fixing the underlying pmf to $p$, i.e.,
	\begin{equation}
		\biasn(p,\tau)=\expectind{\D_n \sim p}{\tau(\D_n)}-\NMI_{p}(X;Y) \enspace .
		\label{eq:biasp}
	\end{equation}
	To estimate the expected value $\mu_n(\tau)$ and standard variation $\sigma_n(\tau)$ of the absolute bias $\abs{\biasn(p,\tau)}$ for the pmfs from $P$, we uniformly sample $100$ pmfs $p^{(1)},\dots, p^{(100)}$ in equal proportions from the four different regimes, i.e., $25$ each from $P[0,0.25]$, $P[0.25,0.5]$, $P[0.5,0.75]$, and $P[0.75,1]$.
	For each pmf $p^{(i)}$ we can calculate the true $\NMI$ value directly from its definition.
	To compute $\biasn(p^{(i)},\tau)$ it then only remains to empirically estimate the expectation term in Eq.~\eqref{eq:biasp}, for which we sample per pmf $p^{(i)}$ a total of $1000$ datasets $D_n \sim p^{(i)}$ of size $n$. By averaging over all $p^{(i)}$, we end up with the desired estimates $\hat{\mu}_n(\tau)$ and $\hat{\sigma}_n(\tau)$ for the absolute bias of estimator $\tau$ with sample size $n$.
	
	\begin{figure}[t]
		\centering
		\begin{minipage}[b]{0.75\linewidth}
			\ifgenplot
			\ifpdf
			\tikzsetnextfilename{figure2}
			\fi
			\begin{tikzpicture}
			\begin{axis}[eda line, xlabel={data size $n$},ylabel={avg. abs. bias $\hat{\mu}_n(\tau)$}, ymin=0.0, ymax=0.5, height=5cm, width=\linewidth, xmin=5.0, xmax=60.0, legend style={nodes={scale=0.8, transform shape}, at={(0.925,1.0)}, anchor=north}, ytick={0.0,0.1,0.2,...,1.0}, xtick={5.0,10.0,20.0,30.0,40.0,50.0,60.0}]
			\pgfplotsinvokeforeach{1,...,3}{
				\addplot table[x index=0, y index=#1, header=true] {../expres/avg_abs_bias.dat};
			}
			\legend{\ourScore, \AMI, \hNMI}
			\end{axis}
			\end{tikzpicture}
			\else
			\includegraphics{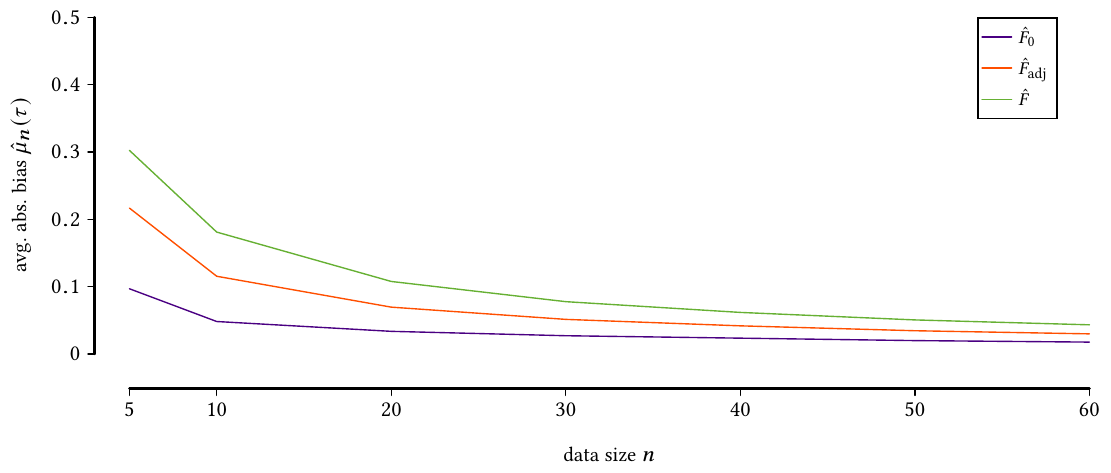}
			\fi 	
		\end{minipage}
		\caption{Estimated absolute bias $\hat{\mu}_n(\tau)$ of estimators $\tau=\hat{\NMI}$, $\AMI$, and the reliable fraction of information $\ourScore$, averaged over all 100 pmfs $p \in P$ across different data sizes $n$.}
		\label{fig:bias}
	\end{figure}
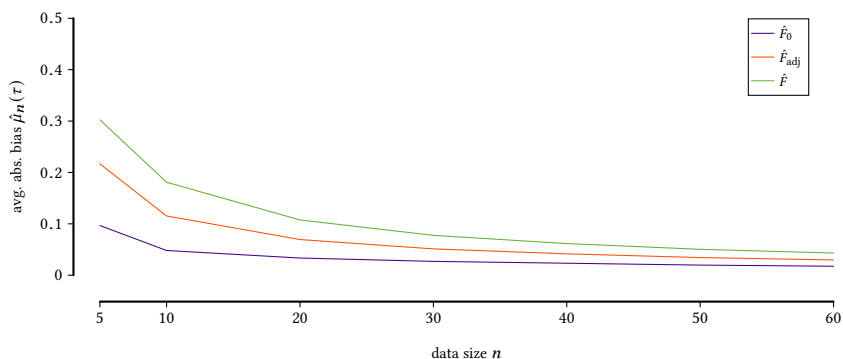
	Equipped with this procedure we can go ahead and compare the performance of $\ourScore$ to other estimators. In addition to the naive uncorrected estimator $\hat{\NMI}$ we also introduce an alternative correction resulting from the application of the quantification adjustment framework proposed by Romano et al.~\cite{romano:2016:chance}.
	This correction, which we denote by \AMI, is defined as
	\begin{align*}
		\AMI(\X;Y)= \frac{\hat{\MI}(\X,Y)-\epm{\hI(\X,Y)}}{\hat{H}(Y)-\epm{\hI(\X,Y)}} \enspace .
	\end{align*}
	Thus, we consider $\tau \in \lbrace \hat{\NMI}, \AMI, \ourScore \rbrace$.
	For each estimator, we compute $\hat{\mu}_n(\tau)$ and $\hat{\sigma}_n(\tau)$ for data sizes $n \in \{5,10,20,30,40,50,60\}$. We focus on small data sizes, because any consistent estimator converges to $\NMI$ for $n \rightarrow \infty$. Furthermore, we can expect the small data sizes for the small domain size $\card{\domain{X}}=3$ of this experiment to behave similar to larger data sizes combined with the potentially huge domains occurring during the algorithmic search (resulting from the complex random variables $\X$).
	
	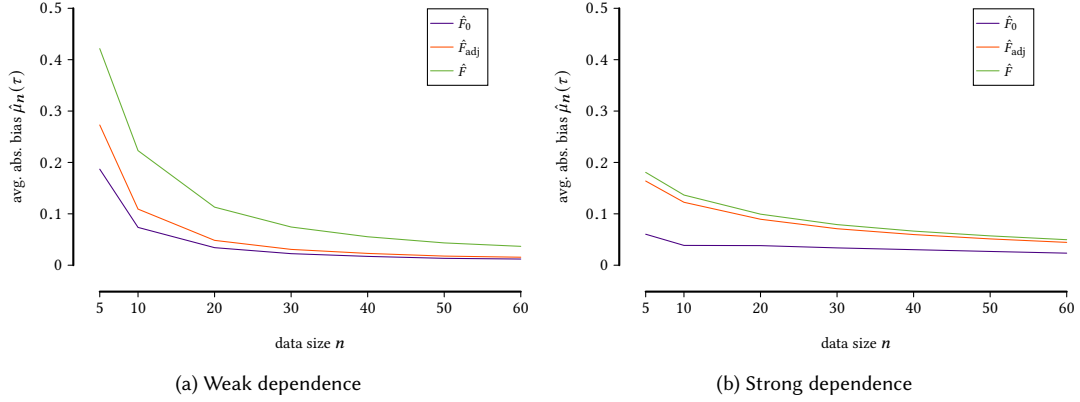
\begin{figure}[t]
		\centering
		\begin{minipage}[b]{0.475\linewidth}
			\centering
			\ifgenplot
			\ifpdf
			\tikzsetnextfilename{figure3}
			\fi
			\begin{tikzpicture}
			\begin{axis}[eda line, xlabel={data size $n$},ylabel={avg. abs. bias $\hat{\mu}_n(\tau)$}, ymin=0.0, ymax=0.5, height=5cm, width=\linewidth, xmin=5.0, xmax=60.0, legend style={nodes={scale=0.8, transform shape}, at={(0.85,1.0)}, anchor=north}, ytick={0.0,0.1,0.2,...,1.0}, xtick={5.0,10.0,20.0,30.0,40.0,50.0,60.0}]
			\pgfplotsinvokeforeach{1,...,3}{
				\addplot table[x index=0, y index=#1, header=true] {../expres/avg_abs_bias_q1.dat};
			}
			\legend{\ourScore, \AMI, \hNMI}
			\end{axis}
			\end{tikzpicture}
			\else
			\includegraphics{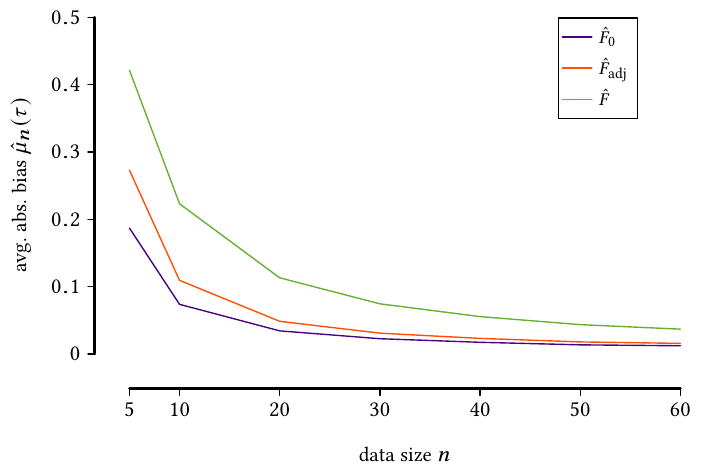}
			\fi 	
			\subcaption{Weak dependence}\label{fig:bias:q1}
		\end{minipage}
		\begin{minipage}[b]{0.475\linewidth}
			\centering
			\ifgenplot
			\ifpdf
			\tikzsetnextfilename{figure4}
			\fi
			\begin{tikzpicture}
			\begin{axis}[eda line,  xlabel={data size $n$},ylabel={avg. abs. bias $\hat{\mu}_n(\tau)$}, ymin=0.0, ymax=0.5, height=5cm, width=\linewidth, xmin=5.0, xmax=60.0, legend style={nodes={scale=0.8, transform shape}, at={(0.85,1.0)}, anchor=north}, ytick={0.0,0.1,0.2,...,1.0}, xtick={5.0,10.0,20.0,30.0,40.0,50.0,60.0}]
			\pgfplotsinvokeforeach{1,...,3}{
				\addplot table[x index=0, y index=#1, header=true] {../expres/avg_abs_bias_q4.dat};
			}
			\legend{\ourScore, \AMI, \hNMI}
			\end{axis}
			\end{tikzpicture}
			\else
			\includegraphics{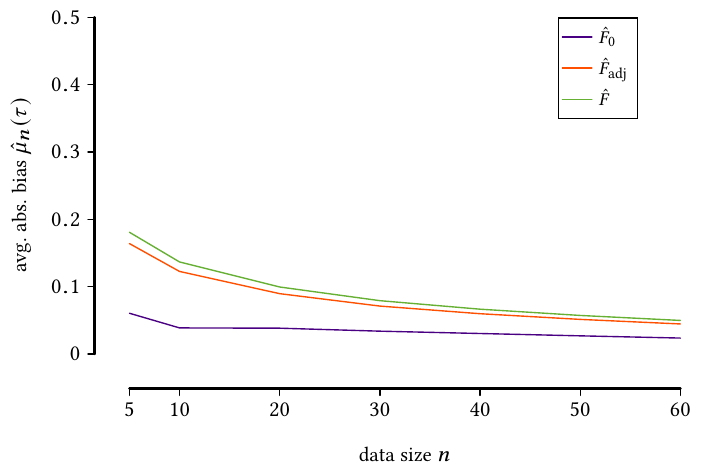}
			\fi 	
			\subcaption{Strong dependence}\label{fig:bias:q4}
		\end{minipage}
		\caption{Estimated absolute bias $\hat{\mu}_n(\tau)$ of estimators $\tau=\hat{\NMI}$, $\AMI$, and the reliable fraction of information $\ourScore$, averaged over $p \in P[0,0.25]$ (left) and over $p \in P[0.75,1]$ (right), across different data sizes $n$.}
		\label{fig:bias2}
	\end{figure}

	In Figure~\ref{fig:bias} we present the results for the estimated absolute bias $\hat{\mu}_n(\tau)$ of estimators $\tau=\hat{\NMI}$, $\AMI$, and the reliable fraction of information $\ourScore$, averaged over all 100 pmfs $p \in P$ across different data sizes $n$. We observe that our estimator $\ourScore$ achieves a lower bias for all $n$ compared to $\hat{\NMI}$ and $\AMI$, and converges fast to a bias close to zero after 10 samples.  The differences in bias are apparent in the cases of $n=5$ and $n=10$ samples. These are cases where the insufficient data samples cause $\hat{H}(Y \mid \X=\vx)$ to approach 0, independent of the true distribution $\hat{p}$. In such scenarios, the estimators $\hat{\NMI}$ and $\AMI$  start to show functional dependence, while $\ourScore$ is designed to show independence for reliability reasons. So it is useful to see that this design, also offers a better bias.
	
	We can further draw conclusions about the behavior of $\ourScore$ by considering only the ``weak" and ``strong" dependencies, i.e., where $\NMI$ is closer to independence and functional dependence respectively. As such, we present in Figure~\ref{fig:bias2} the estimated absolute bias $\hat{\mu}_n(\tau)$, averaged over $p \in P[0,0.25]$ (left), and $p \in P[0.75,1]$ (right). We see that in both cases, i.e., when there is low and high functional dependence, $\ourScore$ achieves a lower bias, as it was the case with $p \in P[0,1]$. Since our score aims to be unbiased under the null hypothesis, we observe a high correction over the $\hat{\NMI}$ for weak dependencies. Even for high dependencies, where one could expect to have less correction, we see that $\ourScore$ is practically unbiased across all $n$.
	
	For the standard deviation $\hat{\mu}_n(\tau)$, we present in Figure~\ref{fig:std} the results after averaging over all 100 pmfs $p \in P$ across different data sizes $n$. We observe that $\ourScore$ has an almost equal $\hat{\mu}_n$ for all $n$ in comparison with $\hat{\NMI}$, and a lower one with $\AMI$. Estimators achieve better bias by trading variance, and from Figures~\ref{fig:bias},~\ref{fig:bias2}, and \ref{fig:std}, we see that in comparison with $\hat{\NMI}$ and $\AMI$, we trade very little variance for a large bias correction. With the previous observations, we can conclude that $\ourScore$ is a suitable estimator for the fraction of information $\NMI$, as desired.

	\begin{figure}[t]
		\centering
		\begin{minipage}[b]{0.75\linewidth}
			\ifgenplot
			\ifpdf
			\tikzsetnextfilename{figure5}
			\fi
			\begin{tikzpicture}
			\begin{axis}[eda ybar, xlabel={data size $n$},ylabel={avg. std. of bias $\hat{\sigma}_n(\tau)$}, ymin=0.0, ymax=0.35, xtick=data, enlarge x limits=0.125, bar width=0.5em, symbolic x coords={5,10,20,30,40,50,60}, legend pos=south east, height=5cm, width=\linewidth, legend image post style={scale=0.25}, legend style={nodes={scale=0.8, transform shape}, at={(0.85,1.0)},anchor=north},
			legend style = {inner sep=1pt, cells={font=\small}, }, 	  	
			]
			\pgfplotsinvokeforeach{1,...,3}{
				\addplot+[mark=none] table[x index=0, y index=#1, header=true] {../expres/avg_std_bias.dat};
			}
			\legend{\ourScore, \AMI, \hNMI}
			\end{axis}
			\end{tikzpicture} 
			\else
			\includegraphics{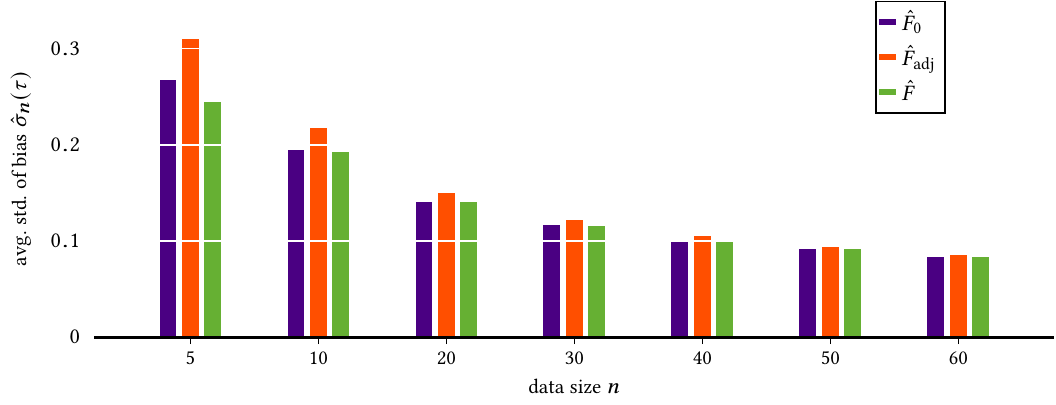}
			\fi 	
		\end{minipage}
		\caption{Estimated standard deviation  $\hat{\sigma}_n(\tau)$ of estimators $\tau=\hat{\NMI}$, $\AMI$, and the reliable fraction of information $\ourScore$, averaged over all 100 pmfs $p \in P$ across different data sizes $n$.}
		\label{fig:std}
	\end{figure}
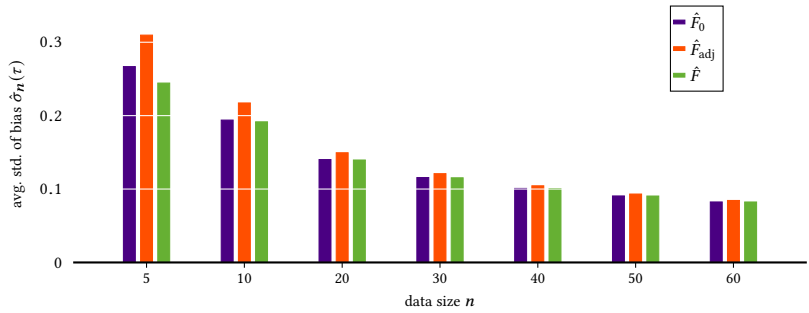

	\subsection{Optimization performance}\label{sec:algPerf}
	
	In this section, we investigate the performance of the branch-and-bound algorithm combined with our optimistic estimator $\oest(\cX)=1-\hat{b}_0(\X,Y,n)$. Specifically, we are interested in the effects of having a data dependent quantity in our bounding function, i.e., $\hat{b}_0(\X,Y,n)$, which in addition, acts as a penalty for non-reliable dependencies. 
	
	For this experiment we utilize the KEEL data repository.\!\footnote{http://www.keel.es/} We use all classification datasets with $n \in [100,13000]$, $d \in [6,90]$, without missing values, resulting in 42 datasets with an average number of 2800 data samples and 24 attributes.\!\footnote{Some attributes with obviously wrong type declarations were corrected, e.g., 'sex' was often corrected to be categoric and not metric.} All metric attributes are discretized using the method of~\citet{fayyad:1993:discr}. The datasets are summarized in Table~\ref{tab:addlabel}. All experiments were executed on a dedicated Intel Xeon E5-2643 v3 machine with 256 GB memory. We make our code available online for research purposes.\!\footnote{\oururl}
	
	We employ the algorithm to retrieve the top dependence, and individually set the $\alpha$ for each dataset, such that the algorithm terminates in less than 30 minutes. We report time, $\alpha$ used, pruning percentage of the search space, depth of the solution, and the maximum depth the algorithm had to explore in Table~\ref{tab:addlabel}.
	
	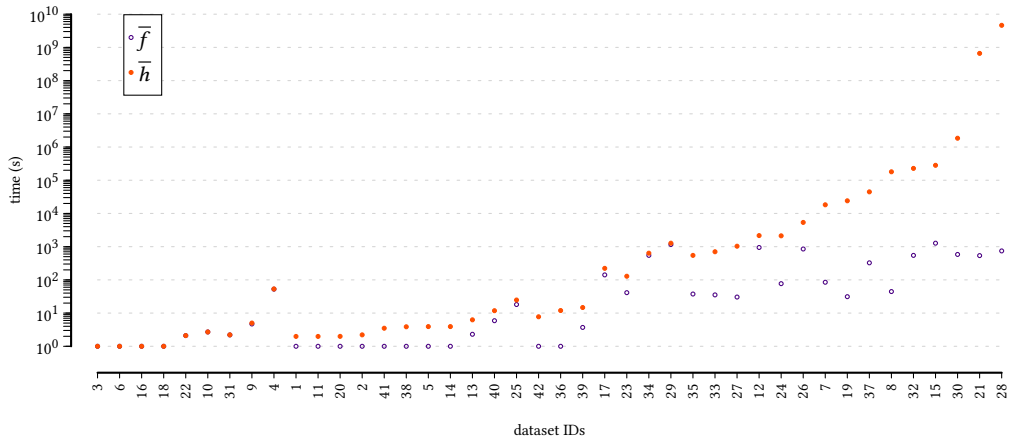
\begin{figure}[t]
		\centering
		\begin{minipage}[b]{0.90\linewidth}
			\ifgenplot
			\ifpdf
			\tikzsetnextfilename{figure6}
			\fi
			\begin{tikzpicture}
			\begin{axis}[
			eda scatter3, 
			mark=*,
			xlabel={dataset IDs}, 
			ylabel={time (s)}, 
			xticklabel style={rotate=90},
			width=\linewidth, 
			legend style={nodes={scale=1, transform shape}, at={(0.05,1)}, anchor=north},
			ymode=log,
			ymin=0.0, 
			ymax=10000000000.0,
			height=6cm,
			ytick={0,1,10,100,1000,10000,100000,1000000,10000000,100000000,1000000000,10000000000,100000000000},
			xtick=data,
			grid = major,
			major grid style  = {lightgray, dash pattern = on 1pt off 3 pt},
			xmajorgrids = false,
			xticklabels={3,6,16,18,1,11,20,22,31,2,10,41,38,5,14,9,13,42,40,36,39,25,4,23,17,35,34,33,27,29,24,12,26,7,19,37,8,32,15,30,21,28},
			legend style = {inner sep=2pt, cells={font=\small}, }, 	 	
			]
			\addplot+[mambacolor1,mark=o,mark size=0.75] table[x index=6,  y index=1] {../expres/runtimes_diff-june17.dat};
			\addplot+[mambacolor2,mark=*,mark size=0.75] table[x index=6, y index=3] {../expres/runtimes_diff-june17.dat};
			\legend{$\oest$, $\hoest$}
			\end{axis}
			\end{tikzpicture} 	 	
			\else
			\includegraphics{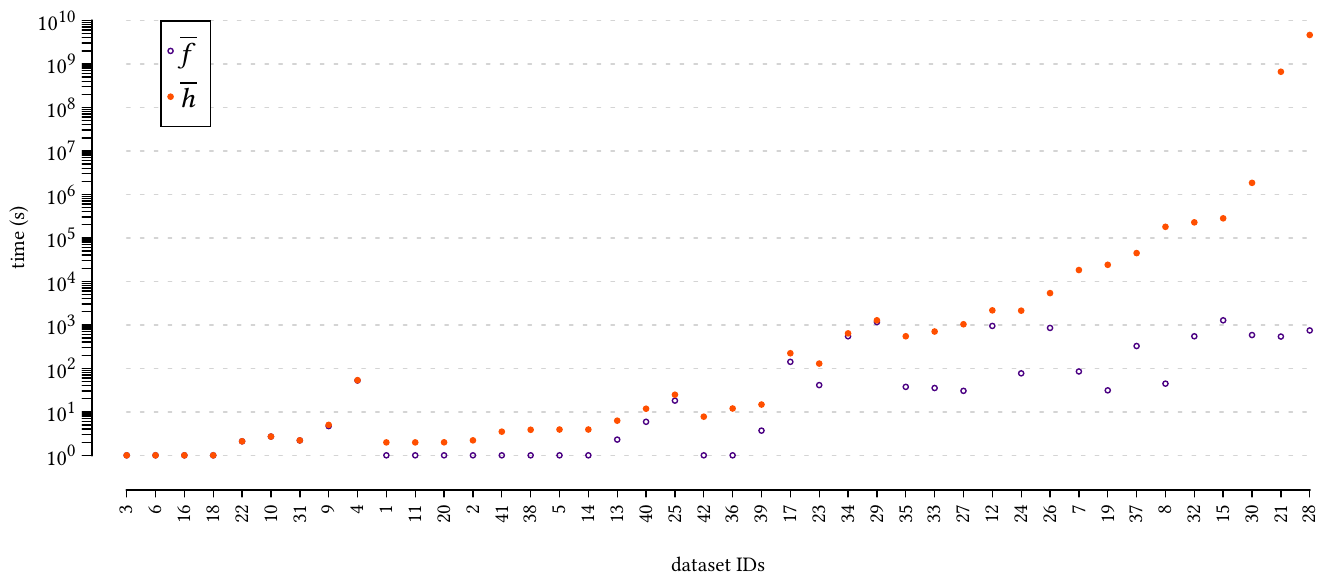}
			\fi 	
		\end{minipage}
		\caption{Computation time of branch-and-bound with $\oest$ and $\hoest$ as bounding function respectively, across all datasets  sorted according to speed-up.}
		\label{fig:algoPerf}
	\end{figure}
	
	The percentage of the datasets where $\alpha=1$, i.e., 30 out of 42, show that an optimal solution can be discovered in under 30 minutes for the majority of the cases considered. In 28 datasets it takes a maximum of 6 minutes. For the rest, reasonable approximations to the optimal solution can be achieved, e.g., $\alpha=0.9, 0.8$.  
	
	We observe that our bounding function $\oest$, is effective in pruning a considerable amount of the search space, i.e., $67.9\%$ on average. In addition, an average of $7.5$ maximum depth combined with an average solution depth of $4.0$, show that $\oest$ is not simply pruning on set cardinality, but it is a data dependent quantity that selectively explores the search space based on the structure of the data. That is, it can potentially go to higher levels for promising candidates. 
	
	To further corroborate on the previous observation, we consider a hypothetical optimistic estimator $\hoest$ that prunes based solely on the cardinality  of $\X$. For a meaningful comparison, we provide an oracle to this method and restrict the search space to the maximum depth that $\oest$ had to explore, and not the complete space of size $2^d$ nodes. For example, if the maximum depth is $l$, then the branch-and-bound with $\hoest$ as bounding function will visit $q=\sum_{i=1}^{l}\binom{d}{i}$ nodes. The estimated time for every dataset is then $q \times t$, where $t$ is a node processing time estimated by dividing the completion time of branch-and-bound with $\oest$, with the number of nodes visited. We plot the computation time for both estimators in Figure~\ref{fig:algoPerf}. The datasets are sorted in ascending order of speed-up. 
	
	We see that in the majority of the datasets, taking into account the structure of the data is of crucial importance. This is most evident in the last two datasets, where it would take 20 and 146 years respectively to find the solution based on cardinality alone. This plot shows the potential of a data dependent optimistic estimator, as opposed to a simple function evaluating statistics as cardinality, in a potentially hard optimization problem. 
	
	Regarding reliability, useful conclusions can be drawn from the average solution dimensionality of $4.0$, which is a reasonable number for the size of the data considered. Trying to maximize other estimators for example, such as $\AMI$ or $\hat{\NMI}$, would result in very large dimensionalities. 
	
	\subsection{Exemplary discoveries}
	
	After investigating the statistical properties of $\ourScore$ and its algorithmic performance, we close this section with examples of concrete approximate functional dependencies discovered in two different applications: determining the winner of a tic-tac-toe configuration and predicting the preferred crystal structure of octet binary semi-conductors.
	Both settings are examples of problems where elementary input features are available, but to correctly represent the input/output relation either non-linear models have to be used or---if interpretable models are sought---complex auxiliary feature have to be constructed from the given elementary features.

	\begin{figure}[t]
		\centering
		\begin{minipage}[b]{0.475\linewidth}
			\centering
			\begin{tikzpicture}[
			tick/.style={},
			highlight/.style={red},
			bar/.style={thick},
			scale=1
			]
			\foreach \i in {0,1,2,3} {
				\path[draw,bar] (0,\i) -- (+3,\i);
				\path[draw,bar] (\i,0) -- (\i,3);
			}
			\foreach \x/\y/\i in {
				1/2/4,
				2/1/2,2/3/8,
				3/2/6} {
				\node[tick] at (\x-.5,3-\y+.5) {$X_\i{}$};
			}
			\foreach \x/\y/\i in {
				1/1/1,1/3/7,
				2/2/5,
				3/1/3,3/3/9} {
				\node[tick,highlight] at (\x-.5,3-\y+.5) {$X_\i{}$};
			}
			\end{tikzpicture}
		\end{minipage}
		\begin{minipage}[b]{0.475\linewidth}
			\centering
			\begin{tikzpicture}[
			tick/.style={},
			bar/.style={thick},
			scale=1
			]
			\foreach \i in {0,1,2,3} {
				\path[draw,bar] (0,\i) -- (+3,\i);
				\path[draw,bar] (\i,0) -- (\i,3);
			}
			\foreach \x/\y/\i in {
				1/1/3,1/2/1,1/3/3,
				2/1/1,2/2/4,2/3/1,
				3/1/3,3/2/1,3/3/3} {
				\node[tick] at (\x-.5,3-\y+.5) {\i};
			}
			\end{tikzpicture}
		\end{minipage}
		\caption{Tic-tac-toe board with input variables in corresponding board positions and variables contained in top approximate functional dependency marked in red ({left}); and number of winning combinations each position is involved in ({right}).}
		\label{fig:tic}
	\end{figure}
	
	The tic-tac-toe application \cite{matheus1989constructive} is one of the earliest examples of this complex feature construction problem.
	Tic-tac-toe is a game of two players where each player picks a symbol from $\lbrace x,o \rbrace$ and, taking turns, marks his symbol in an unoccupied cell of a $3\times3$ game board. A player wins the game if he marks 3 consecutive cells in a row, column, or diagonal. A game can end in draw, if the board configuration does not allow for any winning move. 
	The dataset consists of 958 end game, winning configurations (i.e., there are no draws). The 9 input variables $\cI=\lbrace X_1,\dots,X_9 \rbrace$ represent the cells of the board, and can have 3 values $\lbrace{ x,o,b \rbrace}$, where $b$ denotes an empty cell (see Fig.~\ref{fig:tic}). The output variable $Y$ with $\domain{Y}=\lbrace \text{win}, \text{loss} \rbrace$ is the outcome of the game for player $x$. 
	
	Searching for approximate functional dependencies reveals as top pattern with empirical fraction of information $\hat{\NMI}=0.61$ and corrected score $\ourScore=0.45$ the variable set
	\[
	\X=\lbrace X_1, X_3, X_5, X_7, X_9 \rbrace
	\]
	i.e., the four corner cells and the middle one. This is a sensible discovery as these cells correspond exactly to those involved in the highest number of winning combinations (see Fig.~\ref{fig:tic}).
	Knowing the state of these cell provides, therefore, a high amount of information about the outcome of the game. 
	Moreover, removing a variable results in a loss of a considerable amount of information, while adding a variable would provide more information, but also redundancy. That is, the increase of fraction of information would not be higher than the increase of $\hat{b}_0$.

	\begin{figure}
		\centering
		\includegraphics[width=0.6\linewidth]{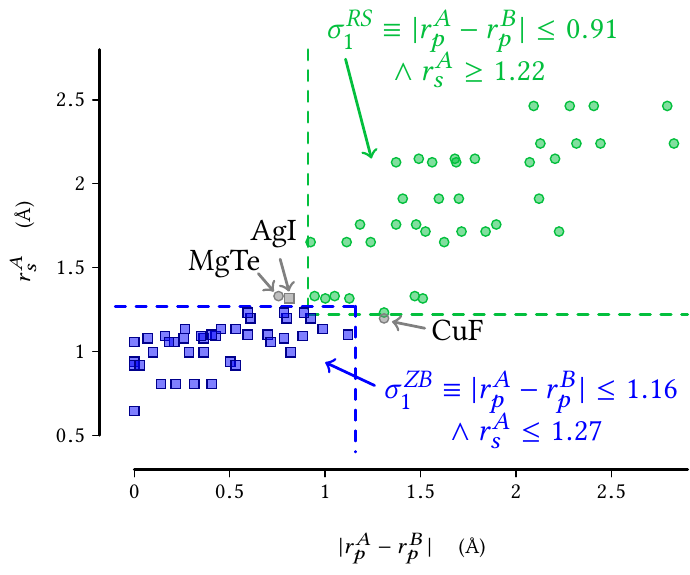}
		\caption{Binary semiconductors that crystalize as zinkblende (boxes) and rocksalt (circles); blue and green materials are correctly classified by subgroup-based prediction model---the involved rules (annotated) use elements of top functional dependency as two out of three variables. (source: \citet{goldsmith2017uncovering}).}
	\end{figure}
	
	Our second example is a classical problem from Materials Science \cite{van1969quantum}, which has meanwhile become a canonical example for the challenge of the automatic discovery of interpretable and ``physically meaningful'' prediction models of material properties \cite{ghiringhelli2015big, goldsmith2017uncovering}.
	The task is to predict the symmetry or crystal structure, in which a given binary compound semi-conductor material will crystalize.
	That is, each of the $82$ material involved consist of two atom types (A and B) and the output variable $Y=\{\text{rocksalt}, \text{zincblende} \}$ describes the crystal structure it prefers energetically.
	The input variables are 14 electro-chemical features of the two atom types considered in isolation: the radii of the three different electron orbitals shapes $s$, $p$, and $d$ of atom type A denoted as $r_s(A), r_p(A), r_d(A)$ as well as four important energy quantities that determine its chemical properties (electron affinity, ionization potential, HOMO and LUMO energy levels); the same variables are defined for component B.
	
	For this dataset the top approximate functional dependency with $\ourScore=0.707$ and uncorrected empirical fraction of information $\hat{\NMI}=0.735$ is
	\[
	\X=\{r_s(A), r_p(A)\}
	\]
	i.e., the atomical $s$ and $p$ radii of component A. Again, this is a sensible finding, since these two variables constitute two out of three variables contained in the best structure prediction model that can be identified using the non-linear subgroup discovery approach \citep{goldsmith2017uncovering}. Also both features are involved in the best linear LASSO model based on systematically constructed non-linear combinations of the elementary input variables \citep{ghiringhelli2015big}.
	The fact that not all variables of those models are identified by the functional dependency discovery algorithm can likely be explained by the facts that (a) the continuous input variables had to be discretized and (b) the dataset is extremely small with only 82 entries, which renders the discovery of reliable patterns with more than two variables very challenging.
	
	\section{conclusion} \label{sec:concl}
	
	We considered the dual problem of measuring and efficiently discovering approximate functional dependencies from data. We adopted an information theoretic approach, and proposed a fraction of information score that is reliable and achieves a good bias. In addition, we proposed an efficient optimistic estimator that allows for the effective discovery of the optimal, or $\alpha$-approximate top-$k$ dependencies of the target variable.
	
	Although we carefully constructed the proposed correction term $\hat{b}_0$ such that bias under those regimes that are most problematic for searching for high-dimensional functional dependencies is removed, other scores could potentially be found that estimate the fraction of information with even less bias.
	
	Other correction terms could also lead to other algorithms. The computational complexity for finding reliable functional dependencies is still open. Hence, polynomial time algorithm for this or adapted problem variants are a possibility. 
	Similarly so, efficient tight(er) optimistic estimators would improve the runtime of branch-and-bound, as fewer nodes would have to be expanded to discover the optimal solution.
	
	Both our score and optimistic estimator are specifically defined for discrete data. While in this paper we only considered univariate discrete targets, our scores can be trivially extended to multivariate discrete variables.
	Clearly, it is also of interest to discover approximate functional dependencies from continuous real-valued data. As entropy has been defined for such data, e.g. differential entropy~\cite{shannon:1948:communication} and cumulative entropy~\cite{rao:04:cre}, it is possible to instantiate fraction of information scores. It will be interesting to see whether we can also efficiently correct these scores for chance, and whether optimistic estimators exist that allow for effective search.

	\begin{table*}[t]
		
		\centering
		\begin{tabular}{rrrrrrrrrr}
			\textbf{ID}    & \textbf{Name} & \textbf{\#rows} & \textbf{\#attrs.} & \textbf{\#clases} & $\alpha$ & \textbf{time(s)} & \textbf{max dep.} & \textbf{sol. dep.} & \textbf{prune $\%$} \\
			\toprule 
			1     & abalone & 4174  & 8     & 28    & 1     & 2.1   & 8     & 3     & 0.00 \\
			2     & appendicitis & 106   & 7     & 2     & 1     & 1.0     & 6     & 3     & 50.00 \\
			3     & tic   & 958   & 9     & 2     & 1     & 1.0     & 7     & 5     & 1.95 \\
			4     & australian & 690   & 14    & 2     & 1     & 3.7   & 11    & 5     & 75.01 \\
			5     & bupa  & 345   & 6     & 2     & 1     & 1.0     & 1     & 1     & 96.88 \\
			6     & car   & 1728  & 6     & 4     & 1     & 1.0    & 5     & 4     & 1.56 \\
			7     & chess & 3196  & 36    & 2     & 0.7   & 84.8  & 7     & 4     & 99.99 \\
			8     & coil2000 & 9822  & 85    & 2     & 0.1   & 44.5  & 5     & 4     & 99.99 \\
			9     & contraceptive & 1473  & 9     & 3     & 1     & 1.0     & 7     & 4     & 75.00 \\
			10     & ecoli & 336   & 7     & 8     & 1     & 1.0     & 6     & 4     & 50.00 \\
			11    & flare & 1066  & 11    & 6     & 1     & 2.7   & 11    & 3     & 0.00 \\
			12    & german & 1000  & 20    & 2     & 1     & 76.9  & 11    & 7     & 97.29 \\
			13    & glass & 214   & 9     & 7     & 1     & 1.0     & 7     & 4     & 75.00 \\
			14    & heart & 270   & 13    & 2     & 1     & 1.0     & 9     & 5     & 75.34 \\
			15    & ionosphere & 351   & 33    & 2     & 1     & 1272.7 & 11    & 4     & 99.98 \\
			16    & led7digit & 500   & 7     & 10    & 1     & 1.0     & 7     & 5     & 0.00 \\
			17    & lymphography & 148   & 18    & 4     & 1     & 41.2  & 11    & 5     & 71.79 \\
			18    & monk  & 432   & 6     & 2     & 1     & 1.0     & 4     & 3     & 10.94 \\
			19    & movement-libras & 360   & 90    & 15    & 0.4   & 31.2  & 5     & 3     & 99.99 \\
			20    & nursery & 12690 & 8     & 5     & 1     & 2.2   & 7     & 5     & 0.78 \\
			21    & optdigits & 5620  & 64    & 10    & 0.5   & 538.1 & 10    & 3     & 99.99 \\
			22    & page  & 5472  & 10    & 5     & 1     & 4.7   & 8     & 4     & 7.23 \\
			23    & penbased & 10992 & 16    & 10    & 1     & 141.2 & 5     & 3     & 93.33 \\
			24    & ring  & 7400  & 20    & 2     & 0.6   & 944.4 & 7     & 3     & 94.24 \\
			25    & saheart & 462   & 9     & 2     & 1     & 1.0     & 5     & 4     & 93.75 \\
			26    & satimage & 6435  & 36    & 7     & 0.9   & 850.9 & 5     & 3     & 99.99 \\
			27    & segment & 2310  & 19    & 7     & 1     & 35.3  & 8     & 2     & 98.37 \\
			28    & sonar & 208   & 60    & 2     & 1     & 744.7 & 13    & 6     & 99.99 \\
			29    & spambase & 4597  & 57    & 2     & 0.5   & 30.4  & 4     & 3     & 99.99 \\
			30    & spectfheart & 267   & 44    & 2     & 0.5   & 583.1 & 10    & 5     & 99.99 \\
			31    & splice & 3190  & 60    & 3     & 0.6   & 52.5  & 3     & 3     & 99.99 \\
			32    & texture & 5500  & 40    & 11    & 0.8   & 546.4 & 7     & 3     & 99.99 \\
			33    & thyroid & 7200  & 21    & 3     & 0.9   & 37.5  & 7     & 3     & 99.35 \\
			34    & twonorm & 7400  & 20    & 2     & 0.9   & 1160.2 & 6     & 4     & 94.74 \\
			35    & vehicle & 846   & 18    & 4     & 1     & 547.2 & 13    & 4     & 15.44 \\
			36    & vowel & 990   & 13    & 11    & 1     & 18.1  & 10    & 3     & 27.69 \\
			37    & wdbc  & 569   & 30    & 2     & 1     & 326.1 & 10    & 4     & 99.96 \\
			38    & wine  & 178   & 13    & 3     & 1     & 1.0     & 6     & 3     & 77.37 \\
			39    & winequality red & 1599  & 11    & 11    & 1     & 1.0     & 8     & 6     & 87.50 \\
			40    & winequality white & 4898  & 11    & 11    & 1     & 5.9   & 10    & 9     & 50.00 \\
			41    & yeast & 1484  & 8     & 10    & 1     & 1.0    & 7     & 7     & 50.00 \\
			42    & zoo   & 101   & 15    & 7     & 1     & 2.3   & 9     & 5     & 84.41 \\
			\midrule
			Average     &       & 2800.0      &   24.0     &   5.6    &   0.89    &  194.0      & 7.5 & 4.0 & 67.9 \\
			\bottomrule
		\end{tabular}%
		\caption{Datasets used in Section~\ref{sec:algPerf}. The table contains information about the name and ID of the datasets used, number of rows, input variables, and classes, the $\alpha$ used for completion in less than 30 minutes, the time in seconds, the maximum depth of the algorithm, the depth of the best solution, and the percentage of the pruned search space.}
		\label{tab:addlabel}%
	\end{table*}%

	\begin{acks}
		The authors wish to thank Luca Ghiringhelli and Brian Goldsmith for insightful discussions. 
		Panagiotis Mandros is supported by the International Max Planck Research School for Computer Science (IMPRS-CS). 
		The authors are supported by the Cluster of Excellence ``Multimodal Computing and Interaction'' within the Excellence Initiative of the German Federal Government.
	\end{acks}
	
	\balance
	
	\bibliographystyle{ACM-Reference-Format}
	

\end{document}